%% file: RS_DCM18_EPTCS.tex
\documentclass[submission,copyright,creativecommons,british]{eptcs}
\providecommand{\event}{DCM'18} 
\usepackage{underscore}           
\makeatletter

\PassOptionsToPackage{pdftex,dvipsnames}{xcolor}
\RequirePackage{xcolor}

\usepackage[final]{microtype}
\usepackage[shorthands=off,safe=none,math=normal]{babel}

\input{francmaths}

\usepackage{amsthm}

\newcommand{\co}{\mathbin{\circ}}

\iffalse
\newcommand{\cat}[1]{\mathord{\mathscr{#1}}}

\newcommand{\Tape}{\mathord{\mathscr{T\!\!a\!\!p\!e}}}
\newcommand{\commaCat}[2]{#1/#2}
\newcommand{\commaDomain}{\operatorname{dom}}
\newcommand{\commaCodomain}{\operatorname{cod}}
\newcommand{\commaArrowTrans}{\operatorname{mid}}
\newcommand{\commaLift}[1]{#1^{\triangleleft}}
\newcommand{\localAdjoint}[1]{#1_{\triangleleft}}
\else

\newcommand{\cat}[1]{\mathord{\mathbb{#1}}}

\newcommand{\Tape}{\mathord{\mathsf{Tape}}}

\newcommand{\commaCat}[2]{#1/#2}
\newcommand{\commaDomain}{\operatorname{dom}}
\newcommand{\commaCodomain}{\operatorname{cod}}
\newcommand{\commaArrowTrans}{\operatorname{mid}}
\newcommand{\commaLift}[1]{#1^{\triangleleft}}
\newcommand{\localAdjoint}[1]{#1_{\triangleleft}}

\fi

\newcommand{\fcellw}{\square}
\newcommand{\fcellb}{\blacksquare}
\newcommand{\cellw}{\kern-0.1em\square\kern-0.02em}
\newcommand{\cellb}{\kern-0.072em\blacksquare}

\newcommand{\definiendum}[1]{\textbf{#1}}

\newtheorem{theorem}{Theorem}
\newtheorem{proposition}[theorem]{Proposition}
\newtheorem{corollary}[theorem]{Corollary}

\theoremstyle{definition} %

\newtheorem{definition}{Definition}

\newcommand{\fctl}[3]{\begin{tikzcd}[ampersand
    replacement=\&, cramped, sep=small]#1\colon #2 \arrow[r] \& #3 \end{tikzcd}}

\newcommand{\fromto}[2]{\begin{tikzcd}[ampersand
    replacement=\&, cramped, sep=small
]#1\arrow[r]\& #2\end{tikzcd}}

\newcommand{\ccat}{\commaCat}
\newcommand{\dom}{\commaDomain}
\newcommand{\cod}{\commaCodomain}
\newcommand{\cnat}{\commaArrowTrans}
\newcommand{\clift}{\commaLift}
\newcommand{\ladj}{\localAdjoint}
\newcommand{\defin}{\definiendum}

\title{A Category Theoretic Interpretation \\ of Gandy's Principles for Mechanisms}
\author{Joseph Razavi \qquad\qquad Andrea Schalk
	\institute{School of Computer Science\\
		University of Manchester\\
		Manchester, UK}
	\email{joseph.razavi@manchester.ac.uk \qquad\qquad andrea.schalk@manchester.ac.uk}
	}

\def\titlerunning{A Category Theoretic Interpretation of Gandy's Principles}
\def\authorrunning{J.~ Razavi \& A.~ Schalk}

\begin{document}
\maketitle

\begin{abstract} Based on Gandy's principles for models of
computation we give category-theoretic axioms describing locally
deterministic updates to finite objects. Rather than fixing a
particular category of states, we describe what properties such a
category should have. The computation is modelled by a functor that
encodes updating the computation, and we give an abstract account of
such functors. We show that every updating functor satisfying our
conditions is computable. 
\end{abstract}

\section{Introduction}

In a well-known paper \cite{Gandy}, Gandy sets out principles which aim to characterize the possible behaviours of a discrete,
deterministic mechanical computing device which could be realized in
the physical world. Although in Gandy's detailed axiomatization there
are four principles, they can be summarized by two conceptual
insights (as emphasized by, for instance,
\cite{ArrighiCausalGraphs,SiegAbstractParallel}): first, that states
of computation should be finite objects with a bounded amount of
local detail; second, that changes should only propagate with a
bounded velocity, and thus their effects on a given location should
be determined by a neighbourhood of finite size.

Gandy's technical realization of these principles uses axioms which
are set-theoretic in nature, the idea being that every mathematical
description of a computing machine ought to correspond to a
set-theoretical one. As long as this is done sensibly, and the
original model satisfies Gandy's conceptual principles, the resulting
formalization will satisfy Gandy's axioms. Although some later
studies keep within this set-theoretic framework
\cite{SiegAbstractParallel}, other work inspired by Gandy's
principles replaces his arbitrary hereditarily finite sets with
mathematical objects for which one has more direct spatial
intuitions, such as graphs
\cite{ArrighiCausalGraphs,ObtulowiczGPRmachines,SiegKgraphs} or
simplicial surfaces \cite{ArrighiCausalSurfaces}. Indeed, when one
has a concrete model of computation in mind, it is often easier to
conduct detailed investigations into the operation of Gandy's
principles by working with the models directly, in an environment
which respects their extra structure, rather than their set-theoretic
encodings.

In many of these models the dynamics are given in terms of a colimit
of updated versions of certain neighbourhoods of a state. This is
explicit in \cite{Maignan,ObtulowiczGPRmachines}, and suggested by the use of
unions in \cite{ArrighiCausalGraphs,ArrighiCausalSurfaces}. In \cite{Maignan} the colimit to be taken is strikingly similar to the one implied by Proposition 1 of this paper. It is interesting that two quite different intuitions about the meaning of local determinism lead naturally to some of the same categorical structures. Indeed,
if one heuristically reads Gandy's notion of restriction of a state
to a part of a state as analogous to a set of morphisms, then Gandy's
own description strongly suggests a colimit. This is evidence that
the ambient categorical structure can play an explanatory role.  

Rather than describing specific categories whose objects can serve as
states of computation, one might ask what kinds of category could
serve as a setting for Gandy-like models. We know that objects like
graphs and simplicial complexes are characterized by being colimits
of a set of generating objects. We might wonder whether the colimit
description of the dynamics follows from this, given some more direct
axiomatization of local determinism.  

In this paper, we attempt to give just such a description of a class
of categories suitable for describing states of computation in the
style of Gandy, along with a categorical description of locally
deterministic dynamics. We came to this axiomatization in an attempt
to better understand the phenomena of information flow in cellular
automata as described in the first author's thesis
\cite{RazaviPhdThesis}. For this purpose, they seem to be useful for
getting at the relevant aspects of the causal structure, and we hope,
therefore, that they may be of use to workers engaged in similar
enquiry. 

\section{Preliminaries: Comma Categories and Kan Extensions}

In this section we briefly sketch some category theoretic background,
loosely following \cite{LehnerKanExtensions,StreetFibrationsYoneda};
the reader interested in a more comprehensive overview of category
theory is directed to \cite{MaclaneCategoriesWork}.

In order to talk about things happening locally, as is required to
model Gandy's ideas, it is useful to think of a morphism
$\fctl{p}{A}{X}$ in $\cat{C}$ as a `shape $A$ located at position $p$
in $X$'. More generally, one often wants to restrict attention to the
case where the source and target of the morphisms are the output of
two functors $\fctl{F}{\cat{X}}{\cat{C}}$ and
$\fctl{G}{\cat{Y}}{\cat{C}}$. In this paper, we study such morphisms
in the standard setting: the \definiendum{comma
  category}~$\commaCat{F}{G}$. The objects of $\commaCat{F}{G}$ are
triples $(X,p,Y)$ where $X$ is an object of $\cat{X}$, $Y$ is an
object of $\cat{Y}$, and $\fctl{p}{FX}{GY}$ in~$\cat{C}$. A morphism
from $(X,p,Y)$ to $(X',p',Y')$ is a pair $(f,g)$ such that~$Gg\co
p=p'\co Ff$.

We may think of such an equation as telling us that the way
the shape $FX$ is located at position $p$ in $GY$ is compatible
with that of $FX'$ at $p'$ in $GY'$ via $Ff$ and $Gg$. 

The definition of objects of the comma category as triples allows us to define two functors, \(\dom\) and \(\cod\), to \(\cat{X}\) and \(\cat{Y}\) respectively, by projecting onto the first and last components.
The
morphisms in the middle components of these triples assemble themselves into a natural
transformation
\[
  \begin{tikzcd}[row sep=small]
    \ccat{F}{G}\arrow{r}{\dom}\arrow[swap]{d}{\cod} & \cat{X}\arrow{d}{F}
    \arrow[Rightarrow, shorten <= 2ex, shorten >=
    4ex,pos=0.3]{ld}{\cnat} \\
    \cat{Y}\arrow[swap]{r}{G}& \cat{C},
  \end{tikzcd}
\]
and this square has a universal property which defines the comma
category up to isomorphism.

When giving a comma category, if we have a subcategory $\cat{G}$ of
$\cat{C}$, we often write $\cat{G}$ to mean its inclusion into
$\cat{C}$, writing for example $\commaCat{\cat{G}}{\cat{C}}$ to mean
the comma category from the inclusion of $\cat{G}$ to the identity on
$\cat{C}$. Similarly, if $A$ is an object of $\cat{C}$, we may write
$A$ to mean the constant functor from the one-object category to $A$. As a first approximation,
one can think of a comma category $\commaCat{F}{G}$ as a sort of
asymmetric pullback, asking not what $F$ and $G$ have in common, but
what the result of probing $G$ with $F$ is. Certain properties are
then unsurprising. The domain and codomain functors on
$\commaCat{\cat{C}}{\cat{C}}$ both have a section, given by sending
each object to its identity morphism. Moreover, if $F$ has a section,
then so does the codomain functor of~$\commaCat{F}{G}$.

In the introduction, we mentioned that in many examples we encounter
functors which can be calculated by colimits. This may indicate the
presence of a Kan extension, which is defined as follows. Suppose we
have functors $\fctl{F}{\cat{X}}{\cat{D}}$ and $\fctl{G}{\cat{X}}{\cat{C}}$.
A \definiendum{(left)  Kan extension} of $G$ along $F$ is a functor 
$\fctl{L}{\cat{D}}{\cat{C}}$ equipped with a natural transformation
$\fctl{\kappa}{G}{L \co F}$ which is universal in the sense that for
every $\fctl{K}{\cat{D}}{\cat{C}}$ with $\fctl{\alpha}{G}{K \co F}$
there exists a unique mediator $\fctl{\mu}{L}{K}$ such that $\alpha$
is equal to the pasting
\[
  \begin{tikzcd}[row sep=small]
    \cat{X} \arrow{r}{G} \arrow{d}{F} &
    \cat{C}\arrow[equal]{d}\arrow[Rightarrow, shorten <= 1.7ex, shorten
    >=3ex,pos=0.3]{ld}{\kappa}\\
    \cat{D}\arrow{r}{L}\arrow[equal]{d}
    &\cat{C}\arrow[equal]{d}\arrow[Rightarrow, shorten <= 1.7ex,
    shorten >=
    3ex,pos=0.3]{ld}{\mu}\\
    \cat{D}\arrow[swap]{r}{K} &\cat{C}.
  \end{tikzcd}
\]

On first encountering the definition, it is common to
wonder whether when $G = L\co F$, the identity transformation exhibits L as
a Kan extension. One soon realizes, however, that far away from the
image of $F$, the values of $L$ may have little to do with $G$; this
will preclude having the universal property. A natural thing to
demand to improve this situation is that $F$ have a section; one
easily verifies that in this case every functor that post-composes
with $F$ to give $G$ is a Kan extension of the latter along the
former. We are interested in Kan extensions with additional
properties. 

A (left) Kan extension $L$ of $G$ along $F$ as below is called
\defin{absolute} if for all $\fctl{H}{\cat{C}}{\cat{E}}$ the
pasting
\[
  \begin{tikzcd}[row sep=small]
  \cat{X}\arrow[swap]{d}{F}\arrow{r}{G} &\cat{C}\arrow[equal]{d}
  \arrow{r}{H}\arrow[Rightarrow, shorten <= 1.5ex, shorten
    >=3ex,pos=0.3]{ld}{\kappa} &\cat{E}\arrow[equal]{d}\\
  \cat{D}\arrow[swap]{r}{L} &\cat{C}\arrow{r}{H} &\cat{E}
  \end{tikzcd}
\]
is a Kan extension. This means that having a Kan extension for $G$
gives us a Kan extension for any composite of~$G$.
For example, suppose
$\fctl{F}{\cat{C}}{\cat{D}}$, $\fctl{U}{\cat{D}}{\cat{C}}$ with a
natural transformation $\fctl{\eta}{1_{\cat{C}}}{U \co F}$. Then
$\eta$ is the unit of an adjunction with left adjoint $F$ and right
adjoint $U$ if and only if it exhibits $U$ as the Kan extension of
the identity along
$F$ and this extension is absolute. 

A Kan extension $L$ of $G$ along $F$ as above is called
\definiendum{pointwise} if whenever we have a functor
$\fctl{J}{\cat{Y}}{\cat{D}}$ and we consider the comma category
$\commaCat{F}{J}$, the pasting
\[
  \begin{tikzcd}[row sep=small]
    \ccat{F}{J} \arrow[swap]{d}{\cod}\arrow{r}{\dom} & \cat{X}
    \arrow{d}{F}\arrow{r}{G}\arrow[Rightarrow, shorten <= 1.5ex, shorten
    >=4ex,pos=0.3]{ld}{\cnat} &
    \cat{C}\arrow[equal]{d}\arrow[Rightarrow, shorten <= 1.5ex, shorten 
    >=3ex,pos=0.3]{ld}{\kappa} \\
    \cat{Y}\arrow{r}{J}&\cat{D}\arrow{r}{L} &\cat{C}
  \end{tikzcd}
\]
is a Kan extension. We can think of a pointwise Kan extension as
being `locally a Kan extension' in that whenever we `probe' $L$ with
a functor $J$ as above, we still get a Kan extension. One reason
pointwise Kan extensions are useful is that for any object $D$ of
$\cat{D}$ we can compute $LD$ as the colimit of the functor given by
precomposing $G$ with
$\fctl{\commaDomain}{\commaCat{F}{D}}{\cat{X}}$, which means that if
we take the big diagram with one copy of each object $X$ of $\cat{X}$
for every morphism from $\fromto{FX}{D}$, and arrows
between them making commuting triangles, then $D$ is the colimit.

One family of examples is given by the fact that every absolute Kan
extension is pointwise. Another example which one frequently
encounters is the idea of a \definiendum{dense subcategory}. A
subcategory $\cat{G}$ of $\cat{C}$ is called \definiendum{dense} if
and only if the identity on $\cat{C}$ is the pointwise Kan extension
of the inclusion of $\cat{G}$ into $\cat{C}$ along itself, meaning
that every object of $\cat{C}$ is the colimit of all ways of mapping
objects of $\cat{G}$ into it. We think of objects in $\cat{C}$ as
regions generated by the shapes in the dense subcategory~$\cat{G}$.
In the sequel, this is used to capture the notion of `finite detail'.
Since every morphism out of a colimit is uniquely determined by
suitable morphisms out of the constituent parts this also
characterizes the morphisms in $\cat{C}$, which are determined by the
way they act on morphisms out of $\cat{G}$ by postcomposition.
Similar reasoning shows that the domain of every object of
$\commaCat{\cat{C}}{\cat{C}}$ is the colimit of all the domains of
morphisms of $\commaCat{\cat{G}}{\cat{C}}$ into it (to see that one
does not end up with spurious extra copies of objects of $\cat{G}$,
the important observation is that every morphism in
$\commaCat{\cat{C}}{\cat{C}}$ factors as a morphism with an identity
of $\cat{C}$ in its domain, followed by one with an identity of
$\cat{C}$ in its codomain). This implies that the identity natural
transformation corresponding to the commuting diagram 
\[
  \begin{tikzcd}[row sep=small]
    \ccat{\cat{G}}{\cat{C}} \arrow[hookrightarrow]{r}\arrow[hookrightarrow]{d}
    &\ccat{\cat{C}}{\cat{C}} \arrow[equal]{d}\arrow{r}{\dom}
    &\cat{C}\arrow[equal]{d}\\
    \ccat{\cat{C}}{\cat{C}}\arrow[equal]{r}
    &\ccat{\cat{C}}{\cat{C}}\arrow{r}{\dom} &\cat{C}
  \end{tikzcd}
\]
is a pointwise Kan extension. The codomains of morphisms play
little role in this argument, and indeed one can replace
$\commaCat{\cat{G}}{\cat{C}}$ and $\commaCat{\cat{C}}{\cat{C}}$
in the above with $\commaCat{\cat{G}}{U}$ and
$\commaCat{\cat{G}}{U}$ respectively, for any functor $U$ into
$\cat{C}$.

Kan extensions enjoy the property that if one vertically composes a
collection of Kan extension squares, the result is again a Kan
extension. Moreover, if every component of the pasting is absolute or
pointwise, so is the result.

\section{Finite Objects and Local Determinism}

Suppose one has in mind a collection of spatially extended states,
and incomplete parts of states, for a `mechanical' model of
computation, along with a notion for how the various parts fit
together to form larger parts and complete states. These constitute a
category $\cat{C}$. Following Gandy, we stipulate that, if the
model is to be truly mechanical, these states must be finitely big
and have a finite amount of possible local detail. 

Let us first consider the restriction on $\cat{C}$ corresponding to
the objects being `finitely big'. Given two finite combinatorial
objects, we expect that the number of ways in which one fits into the
other to be finite. This means that all hom-sets in $\cat{C}$
should be finite; one says that $\cat{C}$ is `locally finite'. 

Only slightly more subtle is the issue of finite local detail. The
idea is that there should be a finite set of objects such that any
object $X$ is determined by the ways in which these objects map
into $X$. This is exactly the condition described above for a dense
sub-category, so this is the notion we use here.

One example of a locally finite category with a finite dense
subcategory is the category of finite graphs: there are finitely many
graph homomorphisms between any two graphs, and every graph can be
constructed as a colimit of nodes and edges. This is, however, a
somewhat unnatural example for our purposes, since allowing nodes of
arbitrarily high degree means that the notion of `locally finitely
detailed' can not mean `locally' in the usual sense for graphs. It is
more natural to take graphs with a fixed bound on the degree. 

Another example, which we return to below, is the category $\Tape$
whose objects are strings over the alphabet $\{\fcellb,\fcellw\}$,
and where a morphism from $A$ to $B$ is a position in $B$ at
which $A$ occurs as a substring. For example, there are two
morphisms from $\mathord{\fcellb\cellw\cellb\cellb}$ to
$\mathord{\fcellb\cellw\cellb\cellb\cellw\cellb\cellb}$, since it
occurs at positions $0$ and $3$. For technical reasons, it is
useful to assume that the empty string is an initial object,
occurring just once in every string. The full subcategory on the
objects
$\{\fcellb,\fcellw,\fcellb\cellb,\fcellb\cellw,\fcellw\cellb,\fcellw\cellw\}$
is dense, because any string is determined by knowing its consecutive
pairs of letters, and which pairs share individual letters. In a
similar vein, one might think of objects made of a finite number of
tiles (even of, say, a Penrose tiling), generated by gluing together
neighbourhoods where they share tiles.


The issue of local determinism is more interesting. First, we suppose
that the way in which states are updated by the dynamics is given by
a functor $\fctl{U}{\cat{C}}{\cat{C}}$. Functoriality is a reasonable
requirement, meaning that the updated versions of the parts of a
state should fit together in the updated version of that state. It
loosely corresponds to Gandy's idea that the update function in his
model should be `structural', able only to observe the way atomic
parts of the state fit together, not the names we have given them. 

As an example, let us consider the one dimensional cellular automaton
on the alphabet $\{\fcellb,\fcellw\}$ in which a white cell becomes
black if any of its neighbours is black. We can model this as a
functor $\fctl{U}{\Tape}{\Tape}$ which removes the two outer cells in
a string (returning the empty string if there are too few cells) and
updates those in the middle according to the rule. For example
$U(\mathord{\fcellb\cellw\cellw\cellw\cellb\cellw}) =
\mathord{\fcellb\cellw\cellb\cellb}$ Functoriality of $U$
corresponds roughly to the fact that the update can not make use of,
for example, the position of a cell in the sequence to determine its
new colour. 

Now we come to local determinism itself. Suppose we update a state,
and then look at a part of the updated situation. This amounts to
considering a map $\fctl{p}{A}{UX}$ for some $A$ and $X$ in
$\cat{C}$. We want to remember the state $X$ we are updating, so it
is best to consider $p$ as an object of $\commaCat{\cat{C}}{U}$.
Since we postulate that the action of $U$ is locally deterministic,
this local effect must have had some local cause $\fctl{f}{N}{Y}$,
for some $N$ and $Y$ in $\cat{C}$, which explains why $p$ occurs in
$UX$ in the following sense. If we update $f$ to get
$\fctl{Uf}{UN}{U}Y$, we should find $p$ inside the result via a
morphism $\fctl{\eta}{p}{Uf}$ in $\commaCat{\cat{C}}{U}$. Because of
the way the comma category is defined, this implies that there is a
morphism $\fctl{\commaCodomain(\eta)}{X}{Y}$ in $\cat{C}$. We think
of $f$ as a `causal neighbourhood' of $p$. Although we could, if we
wished, choose any causal neighbourhood, it is best to choose one
which comes with a reasoning principle. Since we are supposing that
`changes propagate with a bounded velocity', there ought to be a
minimal causal neighbourhood such that any sequence of changes which
could have contributed to the updated state $p$ must have passed into
it. This amounts to saying that for every other causal neighbourhood, that is for every $\fctl{g}{M}{Z}$ in
$\cat{C}$ with a morphism $\fctl{\gamma}{p}{Ug}$ in
$\commaCat{\cat{C}}{U}$, there is a unique $\fctl{\hat{\gamma}}{f}{g}$
in $\commaCat{\cat{C}}{\cat{C}}$ such that
$\gamma = U\hat{\gamma} \co \eta$.


Since we want this assignment of a causal neighbourhood to vary naturally as we vary $p$, it
amounts to an adjunction; but an adjunction where? The simplest
option is to consider the functor
$\fctl{\commaLift{U}}{\commaCat{\cat{C}}{\cat{C}}}{\commaCat{\cat{C}}{U}}$
which takes 
an object $(X,\fctl{f}{X}{Y},Y)$ of $\commaCat{\cat{C}}{\cat{C}}$ to
$(UX,Uf,Y)$ in $\commaCat{\cat{C}}{U}$. Note that the target
component $Y$ is unchanged. It is for this functor that we demand a
left adjoint, say $\localAdjoint{F}$. 

What does this all mean in our example? Consider the occurrence of
$\mathord{\fcellb\cellb}$ at position $2$ in
$\mathord{\fcellb\cellw\cellb\cellb}$. We must find a causal
neighbourhood which explains why it occurs. We choose the occurrence
of $\fcellw\cellw\cellb\cellw$ at the end of
$\mathord{\fcellb\cellw\cellw\cellw\cellb\cellw}$. This updates
precisely to the part we wanted, and it must be the universal
explanation since anything smaller would only be updated to a single
cell. Now the reader may be troubled by a subtle point: one might
have expected that the explanation for the above should be the
occurrence of $\mathord{\fcellw\cellb}$ at position $3$ in the
input, since we know intuitively that its left-hand square will turn
black. If we had defined $U$ to make use of this knowledge,
however, then we could not have defined $\localAdjoint{F}$ in a
functorial manner, since we would have had
$U(\mathord{\fcellb\cellw\cellb}) = \mathord{\fcellb\cellb\cellb}$.
But then the middle $\fcellb$ could be explained just as well by
$\fcellb\cellw$ on the left as by $\fcellw\cellb$ on the right.
This is why we think in terms of neighbourhoods through which
\emph{all} causal influences on an updated part must have passed.
This issue is related to the problem of `overdetermination' familiar
to philosophers (see e.g. \cite{SchafferOverdetermination}). 

To enforce the bounded speed of propagation, we need to stipulate
that $\localAdjoint{F}$ have a certain finiteness property. For any
given object $A$, there are many possible shapes of causal
neighbourhood for parts of shape $A$, depending on the context in
which we find them. However, since there is a
bound on the size of these, and there can only be finitely many
objects of this size, we know that any part $\fctl{p}{A}{UX}$ has
only finitely many different shapes of possible causal
neighbourhoods. Therefore we demand whenever we have a subcategory
$\cat{S}$ of $\commaCat{\cat{C}}{U}$ whose image under the domain
functor, $\commaDomain[\cat{S}]$, is finite, the domain of its
image under $\localAdjoint{{F}}$, which is $(\commaDomain \co
\localAdjoint{F})[\cat{S}]$, is also finite. 

In our example, consider all morphisms out of the object $\fcellb$
into outputs of $U$. Although this is an infinite set of morphisms
(because there are an infinite number of strings containing
$\fcellb$), each of these occurrences is explained by a
substring in the input of one of the forms
$\mathord{\fcellb\cellb\cellb}$, $\mathord{\fcellb\cellb\cellw}$,
$\mathord{\fcellb\cellw\cellb}$, $\mathord{\fcellb\cellw\cellw}$,
$\mathord{\fcellw\cellb\cellb}$, $\mathord{\fcellw\cellb\cellw}$,
or $\mathord{\fcellw\cellw\cellb}$. If we did not stipulate that
this collection be finite, then different occurrences of $\fcellb$
in different outputs of $U$ could have been explained by substrings
of arbitrary length in the input. This would allow such behaviour as
outputting a single cell which is black if and only if the input
codes a halting Turing machine! 

The example we have given is somewhat restrictive, since the partial
states we mention always shrink, and otherwise do not change shape
very much. This is not a necessary limitation. For instance, we could
introduce special `blank' cells on the ends of the strings which,
rather than being removed by $U$, are duplicated to allow the
working area to grow. More radical changes to the intuitive shape of
parts are possible, but harder to describe. Putting all our
conditions together we get the following.

\begin{definition} A \definiendum{categorical Gandy machine} is an
endofunctor $\fctl{U }{\cat{C}}{\cat{C}}$ where $\cat{C}$ is a
locally finite category with a finite dense subcategory, such that:
  \begin{itemize}
  \item the induced functor
    $\fctl{\commaLift{U}}{\commaCat{\cat{C}}{\cat{C}}}{\commaCat{\cat{C}}{U}}$
    has a left adjoint $\localAdjoint{F}$, and
  \item for all subcategories $\cat{S}$ of
$\commaCat{\cat{C}}{U}$ such that $\commaDomain[\cat{S}]$, is
finite, we have that $(\commaDomain \co \localAdjoint{F})[\cat{S}]$
is also finite.
  \end{itemize}
\end{definition}

We are now in a position to prove that every $U$ satisfying these
conditions is computable. We do this in two steps.

\begin{proposition} Let $\fctl{U}{\cat{C}}{\cat{C}}$ be a categorical
  Gandy machine. Let $\cat{G}$ be a dense subcategory of $\cat{C}$ which gives an induced inclusion
  $\begin{tikzcd}\commaCat{\cat{G}}{U}\arrow[hookrightarrow]{r}&\commaCat{\cat{C}}{U}
  \end{tikzcd}$.
    Then $U$ is the pointwise (left) Kan extension of $\commaDomain$
    along $(\commaDomain \co \localAdjoint{F})$ when both are
    precomposed with this inclusion.
\end{proposition}
\begin{proof} Let \(\eta\) be the unit of the adjunction between \(\clift{U}\) and \(\ladj{F}\), and consider the diagram
\[
  \begin{tikzcd}[row sep=small]
    \ccat{\cat{G}}{U} \arrow[hookrightarrow]{d} \arrow[hookrightarrow]{r}
    &\ccat{\cat{C}}{U} \arrow{r}{\dom}\arrow[equal]{d}
    &\cat{C}\arrow[equal]{d}\\
    \ccat{\cat{C}}{U}\arrow{d}{\ladj{F}} \arrow[equal]{r}
    &\ccat{\cat{C}}{U}\arrow[Rightarrow, shorten <= 1ex, shorten
    >=3ex,pos=0.2]{ld}{\eta}\arrow[equal]{d}\arrow{r}{\dom}
    &\cat{C}\arrow[equal]{d}\\
    \ccat{\cat{C}}{\cat{C}}\arrow{d}{\dom}\arrow{r}{\clift{U}}
    &\ccat{\cat{C}}{U}\arrow{r}{\dom} &\cat{C}\arrow[equal]{d}\\
   \cat{C}\arrow{rr}{U} &&\cat{C}.
  \end{tikzcd}
\]
	
  The top row is a pointwise Kan extension by density of $\cat{G}$.
The middle row is a pointwise Kan extension since its left-hand
square comes from an adjunction. It is thus an absolute Kan
extension, preserved by $\commaDomain$. The bottom row is a
pointwise Kan extension since it commutes and $\commaDomain$ has a
section. Then the result follows by pasting of pointwise Kan
extensions.
\end{proof}

\begin{corollary} Let $\fctl{U}{\cat{C}}{\cat{C}}$ be a categorical
Gandy machine. Suppose that objects of $\cat{C}$ are represented by
diagrams in the finite dense subcategory for which they are the
colimit. Then $U$ is computable.
\end{corollary}
\begin{proof} Suppose for some object $X$ of $\cat{C}$, we want to
  compute $UX$. The above Proposition implies that $UX$ is the
  colimit of all objects which can occur as the domain of an object
  $g$ of $\commaCat{\cat{G}}{U}$ such that the domain of
  $\localAdjoint{F}g$ admits a morphism into $X$. Since $\cat{G}$ is
  finite, $\commaDomain[\commaCat{\cat{G}}{U}] = \cat{G}$ is finite.
  Hence, $\commaDomain[\localAdjoint{F}[\commaCat{\cat{G}}{U}]]$ is
  finite. This is the category in which the diagram whose colimit is
  $UX$ must take its values. Note that, in addition to being finite,
  it is independent of $X$. Therefore, it can be `precomputed' for
  use in evaluating $U$.
  
  In order to evaluate $UX$, all commuting
  triangles with one side in this precomputed category and apex $X$
  must be computed. Since $\cat{C}$ is locally finite, this is a
  finite diagram (and it can actually be computed Since a morphism
  into $X$ is determined by its action on morphisms in $\cat{G}$, and
  one can `try all possibilities' for such actions; if desired, a
  similar idea can be used to ensure we have the `standard' diagram
  whose colimit~is~$X$).
\end{proof}

\section{Discussion and Future Work}

The present work is the first step in a programme to study spatial
models of computation from a categorical perspective. Our original
motivation was to study the flow of information in such models along
the lines of \cite{RazaviPhdThesis}. We are particularly interested
in the question of where in the spatially extended state is stored
which piece of information about the computation. In previous work we
lacked a good axiomatization of suitable structures which we now
provide. We believe this framework to be quite robust, and in this
paper we show how its definition may be thought to arise from Gandy's
principles.

A natural next step is to strengthen the connections with these
principles by investigating whether every example satisfying Gandy's
original axioms can be interpreted in this framework. One obvious
hurdle to overcome is that one has to choose a sensible category of
states. This can not always have Gandy's set of states as objects.
For example, Gandy points out that the set of finite structures for a
first-order logical signature satisfies his axioms. However, if the
signature contains function symbols, then the category of structures,
while locally finite, does not have a finite dense subcategory. For
example, consider a signature with a single function symbol, and
interpret this by successor taken \emph{modulo} $n$ on the set
$\{0,...,n-1\}$. These structures each admit no incoming morphisms
from any other structure, and form an infinite family. The solution
is to consider some notion of partial structure. This corresponds to
adding objects not only for Gandy's states, but also enough parts to
cover at least the ones important for the structure
of the machine.

It may also happen that the definition we give corresponds to a
well-behaved subclass of Gandy's machines (and perhaps to
well-behaved classes of examples of the other models inspired by
Gandy). An example of the limitations of the present model is given
by the functor on the category of finite directed graphs which,
wherever it finds a path of length two $x \to y \to z$ in its
input, adds a direct edge $x \to z$, taking a single step carried
out by the obvious algorithm for transitive closure. This is not an
instance of the present framework, since it will add the diagonal to
a rectangle like 
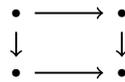
\[ 
  \begin{tikzcd}[row sep=small]
\bullet \arrow{d}\arrow{r} &\bullet\arrow{d}\\
\bullet\arrow{r}&\bullet    
  \end{tikzcd}
\]
but there is no unique smallest causal neighbourhood. This is
related to the problem of overdetermination mentioned above. This
operation is, intuitively, locally deterministic for the usual notion
of `locally' in a graph. However, as we discussed earlier, if nodes
of arbitrarily high degree are allowed, then our definition takes a
strange view of the meaning of `local'.

Comparisons with later models, especially \cite{Maignan}, may be more direct. In \cite{ArrighiCausalGraphs}, a definition of local determinism is given in which a phenomenon in an updated graph is explained by a sub-graph of bounded radius around an `antecedent' in the previous state. This is shown to be equivalent to taking a union of local updates. The analogy between this result and Proposition 1 above raises the question whether the other results of \cite{ArrighiCausalGraphs} have counterparts in the present abstract setting. A different abstract perspective on this idea, based on topology, is presented in \cite{ArrighiCayley1,ArrighiCayley2}; this may aid the comparison with \cite{ArrighiCausalGraphs}. To facilitate this study, a full version of the present work containing more detailed proofs and examples will be presented in a forthcoming paper.

Other important questions are suggested by the key role played by finite categories in Definition 1. It would be interesting if the present extrinsic notion of finite category could be replaced by a more intrinsic condition. In \cite{HaeuslerFiniteness}, the impact on models of computation of varying the notion of finiteness is discussed.

Beyond this, in future work we plan to return to our original aims,
and put this axiomatization to use to give a cleaner account of
\cite{RazaviPhdThesis}, and the flow of information in spatial
models.\\

\textbf{Acknolwedgements.} We would like to thank Pouya Adrom for a
timely pointer to Kan extensions and Francisco Lobo for a
number of very fruitful discussions about the structures we use.


\nocite{*}
\bibliographystyle{eptcs}
\bibliography{rsjournal}
\end{document}

%% file: francmaths.tex
\usepackage{mathtools}
\mathtoolsset{mathic=true}

\usepackage[notext,lcgreekalpha,upint]{stix}
\usepackage[cal=cm,scr=boondoxo]{mathalfa}

\let\circ\vysmwhtcircle

\chardef\fml@textampersand\the\&
\mathchardef\fml@mathampersand\mathcode`\&
\protected\def\&{%
  \ifmmode
    \expandafter\fml@mathampersand
  \else
    \expandafter\fml@textampersand
  \fi}

\newcommand*\obeyampersands{\catcode`\&=\active}
\newcommand*\texampersand{&}
{\obeyampersands\global\let&=\texampersand}
\mathcode`\&="8000\relax
\AtBeginDocument{\obeyampersands}

\protected\def\fml@dollar{\csname\ifmmode)\else(\fi\space\endcsname}
\newcommand*\obeydollars{\catcode`\$=\active}
{\obeydollars\global\let$=\fml@dollar}
\mathcode`\$="8000\relax
\AtBeginDocument{\obeydollars}

\newcommand*\noic{\sb{}\kern-\scriptspace\relax}

\usepackage{ltxcmds}
\usepackage{etoolbox}

\protected\def\fml@opt@sb{%
  \ltx@ifnextchar@nospace[\fml@opt@sb@{}}
\def\fml@opt@sb@[#1]{_{#1}}

\def\do#1{%
  \protected\csdef{#1s}{#1\fml@opt@sb}%
  \protected\csdef{#1rm}{\mathrm#1\fml@opt@sb}%
  \protected\csdef{#1bf}{\mathbf#1\fml@opt@sb}%
  \protected\csdef{#1sf}{\mathsfit#1\fml@opt@sb}
  \protected\csdef{#1bb}{\mathbb#1\fml@opt@sb}%
  \protected\csdef{#1cal}{\mathcal#1\fml@opt@sb}%
  \protected\csdef{#1scr}{\mathscr#1\fml@opt@sb}%
  \protected\csdef{#1frak}{\mathfrak#1\fml@opt@sb}%
}
\forcsvlist\csundef{ss,SS}
\docsvlist{ A,B,C,D,E,F,G,H,I,J,K,L,M,N,O,P,Q,R,S,T,U,V,W,X,Y,Z,
            a,b,c,d,e,f,g,h,i,j,k,l,m,n,o,p,q,r,s,t,u,v,w,x,y,z, }

\def\fml@tmp#1=#2{\ifdef#1{}{\let#1#2}}
\def\do#1{\fml@tmp#1}
\docsvlist{ \Alpha=A, \Beta=B, \Epsilon=E, \Zeta=Z, \Eta=H, \Iota=I, \Kappa=K,
            \Mu=M, \Nu=N, \Omicron=O, \Rho=P, \Tau=T, \Chi=X, \omicron=o, }
\def\do#1{%
  \letcs\fml@tmp{#1}%
  \csletcs{#1}{var#1}%
  \cslet{var#1}\fml@tmp}
\docsvlist{ epsilon, phi, }

\def\fml@tmp#1{\ltx@LocToksA\expandafter{\ltx@gobble#1}}
\def\do#1{%
  \expandafter\fml@tmp\expandafter{\string#1}%
  \protected\csdef{\the\ltx@LocToksA s}{\mathit#1\fml@opt@sb}%
  \protected\csdef{\the\ltx@LocToksA rm}{\mathrm#1\fml@opt@sb}%
  \protected\csdef{\the\ltx@LocToksA bf}{\mathbf#1\fml@opt@sb}%
  \protected\csdef{\the\ltx@LocToksA sf}{\mathsfit#1\fml@opt@sb}
}
\docsvlist{ \Alpha, \Beta,    \Gamma,   \Delta, \Epsilon, \Zeta,
            \Eta,   \Theta,   \Iota,    \Kappa, \Lambda,  \Mu,
            \Nu,    \Xi,      \Omicron, \Pi,    \Rho,     \Sigma,
            \Tau,   \Upsilon, \Phi,     \Chi,   \Psi,     \Omega,
            \alpha, \beta,    \gamma,   \delta, \epsilon, \zeta,
            \eta,   \theta,   \iota,    \kappa, \lambda,  \mu,
            \nu,    \xi,      \omicron, \pi,    \rho,     \sigma,
            \tau,   \upsilon, \phi,     \chi,   \psi,     \omega, }
\docsvlist{ \nabla, \partial, \digamma,
            \varepsilon, \vartheta, \varpi, \varrho, \varsigma, \varphi, }

\newmuskip\verythinmuskip
\AtBeginDocument{%
  \thickmuskip=5mu plus 3mu minus 1mu\relax
  \medmuskip=4mu plus 2mu minus 3mu\relax
  \thinmuskip=3mu plus 1mu minus 1mu\relax
  \verythinmuskip=1.5mu plus 1mu minus 1mu\relax}

\newrobustcmd*\tab{\mathchoice
  {\mskip\thickmuskip}
  {\mskip\medmuskip}
  {\mskip\thinmuskip}
  {\mskip\verythinmuskip}}

\usepackage{mleftright}

\newrobustcmd*\st[1][\middle]{%
  \nonscript\tab#1\vert\allowbreak\nonscript\tab\mathopen{}}

\usepackage{tikz-cd}
\usepackage{relsize}

\let\tikz@ensure@dollar@catcode\relax

\usetikzlibrary{arrows}
\usetikzlibrary{shapes.geometric}

\tikzset{%
  loop/.style={%
    every curve to/.append style={every loop},
    to path={\pgfextra{}\tikz@to@curve@path},
  },
  every loop/.append style={looseness=6,min distance=1ex},
  loop right/.append style={out=30,in=-30},
  loop above/.append style={out=120,in=60},
  loop left/.append style={out=210,in=150},
  loop below/.append style={out=300,in=240},
}

\tikzcdset{%
  every label/.style={%
    /tikz/auto,
    /tikz/font={\relsize{-2}},
    /tikz/inner sep=+0.5ex},
}

\tikzcdset{%
  arrow style=tikz,
  tikzcd to/.tip={to},
  tikzcd bar/.tip={|},
  tikzcd left hook/.tip={left hook},
  tikzcd implies/.tip={implies},
}